\documentclass{article}

\usepackage[utf8]{inputenc}
\usepackage[colorlinks = true,
            linkcolor = blue,
            urlcolor  = blue,
            citecolor = blue,
            anchorcolor = blue]{hyperref}
\usepackage{amsmath}
\usepackage{amsthm}
\usepackage{amssymb}
\usepackage{xcolor}
\colorlet{darkgreen}{green!50!black}

\usepackage{algorithm,caption}
\usepackage{algorithmic}
\usepackage{cleveref}
\usepackage{graphicx}

\usepackage{natbib}
\usepackage{framed}
\usepackage{tcolorbox}
\usepackage{times}

\usepackage{geometry}
\usepackage{tikz}
\usetikzlibrary{matrix,arrows,snakes,shapes,trees,positioning}
\tikzset{
    >=latex,
    punkt/.style={
           rectangle,
           rounded corners,
           draw=black, very thick,
           text width=6.5em,
           minimum height=2em,
           text centered},
    pil/.style={
           ->,
           double,
           thick,
           shorten <=2pt,
           shorten >=2pt,},
    punkti/.style={
           rectangle,
           rounded corners,
           draw=black, very thick,
           text width=26.5em,
           minimum height=2em,
           text centered},
    punktii/.style={
           rectangle,
           rounded corners,
           draw=black, very thick,
           text width=18.5em,
           minimum height=2em,
           text centered}
}

\makeatletter
\theoremstyle{plain}
\newtheorem{theorem}{\protect\theoremname}
\newtheorem*{theorem*}{\protect\theoremname}
\newtheorem*{proposition*}{\protect\theoremname}
\newtheorem{definition}{\protect\definitionname}
\theoremstyle{definition}

\newtheorem{example}[definition]{\protect\examplename}
\theoremstyle{plain}
\newtheorem{lemma}[definition]{\protect\lemmaname}

\newtheorem*{cor*}{\protect\corollaryname}

\newtheorem{proposition}[definition]{\protect\propname}

\newtheorem*{question*}{\protect\questionname}
\newtheorem*{assumption*}{\protect\assumptionname}

\newtheorem{conjecture}{\protect\conjecturename}

\newtheorem{innercustomthm}{Theorem}
\newenvironment{customthm}[1]
  {\renewcommand\theinnercustomthm{#1}\innercustomthm}
  {\endinnercustomthm}

\providecommand{\questionname}{Question}
\providecommand{\assumptionname}{Assumption}
\providecommand{\observationname}{Observation}
\providecommand{\corollaryname}{Corollary}
\providecommand{\definitionname}{Definition}
\providecommand{\lemmaname}{Lemma}
\providecommand{\claimname}{Claim}

\providecommand{\theoremname}{Theorem}
\providecommand{\exercisename}{Exercise}
\providecommand{\examplename}{Example}
\providecommand{\remarkname}{Remark}
\providecommand{\propname}{Proposition}
\providecommand{\conjecturename}{Conjecture}

\newtheorem{open}{Open Question}
\makeatother

\newcommand{\C}{\mathcal{C}}

\newcommand{\F}{\mathcal{F}}
\newcommand{\str}{\mathtt{shatter}}
\newcommand{\vc}{\mathtt{vc}}
\newcommand{\conv}{\mathtt{conv}}
\newcommand{\radon}{\mathtt{r}}
\newcommand{\bu}{\mathtt{bu}}

\newcommand{\bb}[1]{{\textcolor{blue}{#1}}}
\newcommand{\R}{\mathbb{R}}
\newcommand{\Z}{\mathbb{Z}}
\newcommand{\I}{\mathcal{I}}

\title{Dual VC Dimension Obstructs Sample Compression by Embeddings}
\date{May 23, 2024}

\author{Zachary Chase\footnote{Department of Mathematics, Technion. Supported by the European Union (ERC, GENERALIZATION, 101039692).} 
\and Bogdan Chornomaz\footnote{Department of Mathematics, Technion. Supported by the European Union (ERC, GENERALIZATION, 101039692).} 
\and Steve Hanneke\footnote{Department of Computer Science, Purdue University}
\and Shay Moran\footnote{Departments of Mathematics, Computer Science, and Data and Decision Sciences, Technion and Google Research.
Robert J.\ Shillman Fellow; supported by ISF grant 1225/20, by BSF grant 2018385, by an Azrieli Faculty Fellowship, by Israel PBC-VATAT, by the Technion Center for Machine Learning and Intelligent Systems (MLIS), and by the European Union (ERC, GENERALIZATION, 101039692). Views and opinions expressed are however those of the author(s) only and do not necessarily reflect those of the European Union or the European Research Council Executive Agency. Neither the European Union nor the granting authority can be held responsible for them.
} 
\and Amir Yehudayoff\footnote{Department of Mathematics, Technion and Department of Computer Science, Copenhagen University. Supported by the BSF, 
by the Danish National Research Foundation, 
and the Pioneer Centre for AI, DNRF grant number P1.}}

\begin{document}

\maketitle

\begin{abstract}
This work studies embedding of arbitrary VC classes in well-behaved VC classes, focusing particularly on extremal classes.
Our main result expresses an impossibility:
such embeddings necessarily require a significant increase in dimension.
In particular, we prove that for every $d$ there is a class with VC dimension $d$
that cannot be embedded in any extremal class of VC dimension smaller than exponential in $d$.

In addition to its independent interest, this result has an important implication in learning theory, as it reveals a fundamental limitation of one of the most extensively studied approaches to tackling the long-standing sample compression conjecture \citep{Warmuth:03}. 
Concretely, the approach proposed by \cite{floyd:95} entails embedding any given VC class into an extremal class of a comparable dimension, and then applying an optimal sample compression scheme for extremal classes. However, our results imply that this strategy would in some cases result in a sample compression scheme at least exponentially larger than what is predicted by the sample compression conjecture.

The above implications follow from a general result we prove: any extremal 
class with VC dimension $d$ has dual VC dimension at most $2d+1$. 
This bound is exponentially smaller than the classical bound $2^{d+1}-1$ of \cite{assouad:83}, 
which applies to general concept classes (and is known to be unimprovable for some classes).
We in fact prove a stronger result, establishing that $2d+1$ upper bounds the dual Radon number of extremal classes.
This theorem represents an abstraction of the classical Radon theorem for convex sets, 
extending its applicability to a wider combinatorial framework, without relying on the specifics of Euclidean convexity.
The proof utilizes the topological method and is primarily based on variants of the Topological Radon Theorem~\citep{Bajmoczy:79}.
\end{abstract}


\section{Introduction and Main Results}

A common idea throughout mathematics is the one of transforming an abstract object of interest into a related ``nice'' object, possessing useful properties.
For example, transforming general functions into smooth functions, 
transforming arbitrary graphs into connected graphs,
or transforming geometric bodies into convex sets.
In algebraic geometry, a classical example is the resolution of singularities where the goal is to ``make varieties smooth''.
Quantitatively, it is often important to understand how much we need to ``add'' 
to an arbitrary object to make it ``nice''.
For instance, 
we can ask how many edges we need to add to a graph to make it connected,
or how much volume we must add to a body to make it convex. 

This general idea has also come up in various contexts within learning theory. 
For example the use of surrogate or regularized losses to accelerate optimization algorithms. 
Another example, which arises in the context of sample compressions concerns embedding VC classes 
within \emph{extremal} classes (a.k.a.\ lopsided, ample, or shattering-extremal; see e.g.~\cite{lawrence:83,Dress:96,Moran:12}).
Extremal classes enjoy many properties which are useful in the context of learning theory, 
for example, for defining sample compression schemes of size equal the VC dimension \cite{moran2016labeled}, 
or proper optimal PAC learners \cite{bousquet:20}. 
For this reason, there has been much interest in transforming arbitrary concept classes 
into extremal classes, and understanding how much one needs to add to an arbitrary class to make it extremal.
Concretely, determining how much the VC dimension must increase (defined below).

This question has been studied in many works (e.g., \citealp{floyd:95,chepoi:20,chepoi:22,rubinstein:15,Rubinstein:22}, and numerous references below), but it has remained open whether it is always 
possible to avoid a significant increase in VC dimension (e.g., a universal constant factor).
The question has been of particular interest in the literature on sample compression schemes, 
since if it were possible to embed arbitrary concept classes into extremal classes without significantly 
increasing the VC dimension, it would immediately imply a positive resolution of the long-standing 
\emph{sample compression conjecture} \citep{Warmuth:03}.

In the present work, we prove that for some concept classes $\C$,
it is necessary to add \emph{quite a lot} to make it extremal; 
that is, any extremal class containing $\C$ must have 
\emph{exponentially larger} VC dimension.
More generally, for any class $\C$, 
we argue that the increase in VC dimension must be at least 
proportional to the \emph{Radon number} of version spaces 
(which is never smaller than the dual VC dimension).
In particular, our result implies that the commonly-studied 
approach aiming to resolve the sample compression conjecture by embedding in extremal classes
cannot significantly improve over the size of known sample compression schemes \citep*{moranyehudayoof:16}.

\subsection{The Sauer-Shelah-Perles and Pajor Inequalities}
The Sauer-Shelah-Perles\footnote{The result was independently proven by \cite{vapnik:68,sauer:72,Shelah:72}, with \cite{vapnik:68} having proven a slightly weaker version. \cite{Shelah:72} gives credit also to Micha Perles. Word has it that, amusingly, Perles proved the result twice, ten years apart, 
and with different proofs! This unique contribution has humorously led some to refer to it as the Perles-Sauer-Shelah-Perles Lemma.} (SSP) inequality is one of the most basic results in the Vapnik-Chervonenkis (VC) theory. 
It is a key technical lemma in the proof of the fundamental equivalence between finite VC dimension, uniform laws of large numbers, and PAC learnability. 
The SSP inequality also plays a significant role in geometry, combinatorics, and model theory.

\begin{definition}[VC Dimension]
A concept class $\C \subseteq \{0,1\}^n$ is said to \emph{shatter} a set $\{x_1, \ldots, x_k\} \subseteq [n]$ if $$\{(c(x_1), \ldots, c(x_k)): c \in \C\} = \{0,1\}^k$$ (by convention, the empty set is always shattered). The VC dimension of $\C$, denoted by $\mathtt{vc}(\C)$, is the largest size of a set that is shattered by $\C$.
\end{definition}

\begin{theorem}[Sauer-Shelah-Perles Inequality \citep{sauer:72, Shelah:72}]\label{t:sauer}
    Let $\C\subseteq\{0,1\}^n$ be a concept class, and let $d = \mathtt{vc}(\C)$. Then,
    \[\lvert\C\rvert \leq {n \choose \leq d} := \sum_{i=0}^d {n \choose i}.\]
\end{theorem}
The SSP inequality is tight: for every $n$ and $d$, there exist classes that achieve equality in it.
These classes are interesting and have a relatively rigid structure. 
They also arise naturally in geometrically defined classes such as halfspaces in $\mathbb{R}^n$~\citep{Gartner:94}.

\begin{definition}[Maximum Classes]
A class $\C\subseteq\{0,1\}^n$ is called a \emph{maximum} class if \(\lvert\C\rvert = {n \choose \leq d}\), where $d = \mathtt{vc}(\C)$.
\end{definition}

A notable generalization of the SSP inequality was proven by~\cite{Pajor:1985}:
\begin{theorem}[Pajor Inequality \citep{Pajor:1985}]\label{t:pajor}
    Let $\C\subseteq\{0,1\}^n$ be a concept class, and let 
    \[\str(\C) = \{A\subseteq [n]: A\text{ is shattered by } \C\}.\] 
    Then,
    \[\lvert\C\rvert \leq  \lvert \str(\C)\rvert.\]
\end{theorem}
Pajor's inequality indeed implies the SSP inequality, because
the size of $A \in \str(\C)$
is $|A| \leq \mathtt{vc}(\C)$.
Pajor's inequality is also tight. Classes that meet Pajor's inequality with equality are known as lopsided~\citep{lawrence:83}, ample~\citep{Dress:96}, or shattering-extremal~\citep{Moran:12}. In this text, we use the term \emph{extremal} \citep{Moran:12,moran2016labeled}.
\begin{definition}[Extremal Classes]
A class $\C\subseteq\{0,1\}^n$ is called \emph{extremal} if \(\lvert\C\rvert = \lvert\str(\C)\rvert\).
\end{definition}
While every maximum class is extremal, the converse does not hold.

Extremal classes are notable for their characterization through a range of equivalent definitions, which may initially seem quite distinct. For example, \cite{lawrence:83} discovered extremal classes in his investigation of sign patterns of convex sets: $\{\mathtt{sign}(v): v \in K\}$, where $K \subseteq \mathbb{R}^n$ is convex. He defined these classes in a manner that is unrelated to \Cref{t:pajor} and does not even rely on the notion of \emph{shattering}. Consequently, these classes have been independently identified in various contexts, including discrete geometry~\citep{lawrence:83}, functional analysis~\citep{Pajor:1985}, extremal combinatorics~\citep*{Bollobas:89, Bollobas:95}, and computational biology~\citep{Dress:96}.

\subsection{Sample Compression Schemes}

Sample compression schemes were introduced by \cite{littlestone:86} in their seminal paper as a tool for proving generalization bounds in learning theory.
 More generally, this concept can be understood as a mathematical model for data simplification. It is akin to a scientist who collects a large number of experimental observations and then selects a small, representative subset. The scientist's aim is to develop a hypothesis from this subset that can accurately predict the outcomes of the remaining, unselected observations.

A sample compression scheme comprises a {\it compressor} and a {\it reconstructor} (see Figure~\ref{fig:comp}).
The compressor gets an input sequence of labeled examples $S$.
From this sequence, the compressor selects a subsequence $S'$ and sends it to the reconstructor, along with a binary string $B$ of additional information.
The reconstructor then uses $S'$ and $B$ to generate a hypothesis~$h=h(S',B)$.
The aim is that the hypothesis $h$ correctly classifies the entire input sequence $S$, including on those examples in $S$ that are not transmitted to the reconstructor.

\begin{definition}
A sample compression scheme for a concept class~\(\mathcal{C}\) is defined by the following criterion. For any input sequence \(S=\{(x_i,y_i)\}_{i=1}^m\) that is realizable\footnote{I.e.\ there exists $c\in \C$ such that \(c(x_i)=y_i\) for all $i\leq m$.} by \(\mathcal{C}\), the hypothesis \(h = h(S', B)\) generated by the reconstructor must be consistent with the entire input sample \(S\); i.e.\ \(h(x_i)=y_i\) for all $i \in [m]$.
We say that the size of the compression scheme is $s$ if the number of bits in $B$
plus the size of $S'$
is at most $s$.
\end{definition}

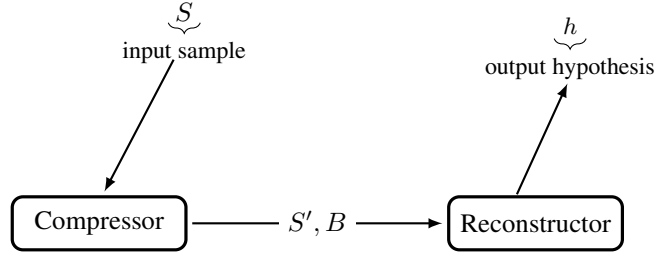
\begin{figure}
\centering
    \textbf{A pictorial definition of a sample compression scheme}\par\bigskip   
\begin{tikzpicture}[scale=0.7]

\node[inner sep=0pt,punkt] (bob) at (7.1,0) {Reconstructor};

\node[inner sep=0pt,punkt] (alice) at (-1.1,0){Compressor};

\draw[->,thick] (0.65,0) -- (5.35,0)
    node[midway,fill=white] {$S',B$};


\node[inner sep=0pt] (sample) at (0.5,4){$S$};
\draw[->,thick] (0.3,3.1) -- (-1,0.6);

\draw [decorate,decoration={brace,amplitude=4pt,mirror},yshift=0pt]
(0.2,3.8) -- (0.8,3.8) node [black,midway,yshift=-0.4cm]
{\small input sample};

\node[inner sep=0pt] (reconstruction) at (7.75,3.75){\small $h$};

\draw [decorate,decoration={brace,amplitude=4pt,mirror},yshift=0pt]
(7.45,3.5) -- (8.05,3.5) node [black,midway,yshift=-0.4cm]
{\small output hypothesis};

\draw[->,thick] (6.75,0.55) -- (7.7,2.65);
\end{tikzpicture}
\caption{
$S'$ is a subsample of $S$ and $B$ is a binary string of additional information.}
\label{fig:comp}
\end{figure}

A classical example of a sample compression scheme for the class of $d$-dimensional linear classifiers is the {\it Support Vector Machine} algorithm in $\mathbb{R}^d$. 
Here, the compressor sends to the reconstructor the~$d+1$ support vectors which determine the maximum margin separating hyperplane. 

Many well-studied classes admit a sample compression of size equal to their VC dimension.
\cite{littlestone:86} asked
about a general connection between the VC dimension and the size of the sample compression scheme.
This became one of the most extensively studied and long-standing problems in the theory of learning. 
\cite{moranyehudayoof:16} affirmatively demonstrated the existence of a sample compression scheme with size exponential in the VC dimension. Yet, it remains an open question whether every class has a compression scheme of size linear or polynomial in its VC dimension, as explicitly asked in \cite{floyd:95}. \cite{Warmuth:03} has offered a \(\$600\) reward for solving this problem. For a broader discussion, see e.g.\ Chapter~17 in \cite{Wigderson:17} or Chapter~30 in \cite{Shaelv-Shwartz:book}.

\begin{conjecture}[\cite{floyd:95,Warmuth:03}]\label[Conjecture]{con:scs}
For every concept class $\C$ there exists a sample compression scheme of size $O(\mathtt{vc}(\C))$.
\end{conjecture}

In their important work, \cite{floyd:95} layed an approach aimed at resolving the sample compression conjecture. This approach is based on a simple two-step program.

\begin{framed}
\textbf{Floyd and Warmuth Program}
\begin{enumerate}
\item Prove that any maximum class \(\mathcal{C}\) with VC dimension \(d\) has a sample compression scheme of size \(O(d)\).
\item Prove that any class \(\mathcal{C}\) with VC dimension \(d\) is contained within a maximum class with VC dimension \(O(d)\).
\end{enumerate}
\vspace{-1em}
\end{framed}

The Floyd and Warmuth program has been a prominent method in attempts to resolve Conjecture~\ref{con:scs}. Several studies have explored this approach, including works by \citet*{Ben-David:98,Warmuth:03,Kuzmin:07,Rubinstein:09,Rubinstein:12,rubinstein:15,moran2016labeled, chepoi:20,Chepoi:21,chepoi:22,Chalopin:22,Rubinstein:22,Chalopin:23}.

In their work, \cite{floyd:95} addressed the first step of their program by demonstrating that every maximum class \(\mathcal{C}\) admits a sample compression scheme of size equal to its VC dimension.
Following this, a series of subsequent works expanded on the scheme proposed by \cite{floyd:95}. 
Specifically, \cite{moran2016labeled} extended it to extremal classes, thereby affirming Conjecture~\ref{con:scs} for this category of classes. This extension not only confirmed the conjecture for extremal classes but also simplified Floyd and Warmuth's program. It demonstrated that the second step of their program could be relaxed, showing that it is sufficient to embed any class in an extremal class while maintaining a comparable VC dimension.


For the second step, \cite{floyd:95} exhibited a class \(\mathcal{C}\) that is maximal (not maximum), meaning that any addition of a concept to this class would increase its VC dimension, but is not maximum. This finding implies that not every class \(\mathcal{C}\) can be extended to a maximum class without increasing the VC dimension.
In the other direction, a recent study by \cite{Rubinstein:22} proved that any intersection-closed class is contained in an extremal class with a VC dimension at most~$11$ times larger.

Our first main result places a barrier on Floyd and Warmuth's program,  proving that it cannot prove Conjecture~\ref{con:scs}.
\begin{framed}
\vspace{-1em} 
\begin{customthm}{A}[Main Result I]\label{t:a}
 For every integer $d>0$, there exists a concept class $\C$ with $\vc(\C)=d$
    such that every extremal class $\C'$ such that $\C\subseteq\C'$ has \(\vc(\C')\geq 2^d - 1\).
\end{customthm}
\vspace{-1em} 
\end{framed}

\Cref{t:a} suggests that the best outcome achievable via Floyd and Warmuth's program would be a sample compression scheme of size \(2^{d}-1\). This would be on par with the best known scheme by \cite{moranyehudayoof:16}, which is of size \( 2^{O(d)}\). However, even reaching this benchmark appears non-trivial,
which brings us back to the first paragraph of this work. 
Can we embed arbitrary classes in ``nice'' classes?
\begin{open}
Is there a function \(f:\mathbb{N}\to\mathbb{N}\) such that every concept class with VC dimension~\(d\) is contained in a maximum (extremal) class with VC dimension \(f(d)\)?
\end{open}

A special case of maximum classes is hyperplane arrangements (see~\Cref{ex:hyperplanes}). 
The smallest dimension of a hyperplane arrangement that contains a class $\C$ is equivalent to the sign rank of the matrix $(M_{c,x})$ defined by $M_{c,x} = c(x)$.
\citet*{alon2017sign} proved that if we replace ``maximum class'' by ``hyperplane arrangment'' in the question, then the answer is negative.
For every $m$, there is a class of VC dimension two that cannot be embedded in an $m$-dimensional hyperplane arrangement.

\subsection{Dual VC Dimension}

We establish \Cref{t:a} by distinguishing extremal classes from general classes with the same VC dimension, focusing on the interplay between VC and dual VC dimensions. 

\begin{definition}[Dual Class and VC dimension]
Let $\C \subseteq \{0,1\}^n$
be a class.
Each $x \in [n]$ defines a function $f_x:\C \to \{0,1\}$ by
$f_x(c) = c(x)$.
The dual class of $\C$ is
$\C^\star= \{f_x : x \in [n]\}$.
The dual VC dimension of $\C$
is $\mathtt{vc}^\star(\C)
= \mathtt{vc}(\C^\star)$.
\end{definition}


The VC and dual VC dimensions are exponentially related:
\begin{theorem}[\cite{assouad:83}]\label{t:assouad}
Every concept class $\C$ satisfies 
\[\lfloor \log \vc(\C)\rfloor \leq \mathtt{vc}^\star(\C) \leq 2^{\mathtt{vc}(\C)+1} - 1.\]
\end{theorem}
Both inequalities are sharp (in some cases). While the relationship between VC and dual VC dimensions is fully understood for general classes, our result provides an exponential improvement for extremal classes
(in the right inequality).

\begin{framed}
\vspace{-1em} 
\begin{customthm}{B}[Main Result II]\label{t:b}
 Every extremal concept class $\C$ satisfies 
\[\lfloor \log \vc(\C)\rfloor \leq \mathtt{vc}^\star(\C) \leq 2\mathtt{vc}(\C)+1.\]
\end{customthm}
\vspace{-1em} 
\end{framed}

\Cref{t:a} is a corollary of \Cref{t:b}. Indeed, choose a concept class $\C$, with $\vc(\C)=d$ and $\vc^\star(\C) = 2^{d+1}-1$. This class can be chosen as a dual of a class $\C_1$, witnessing the sharpness of the right bound in~\Cref{t:assouad}, that is, of $\C_1$ for which $\vc^\star(\C_1)=d$ and $\vc(\C_1)=2^{d+1}-1$.
Let $\C'$ be an extremal class such that $\C'\supseteq \C$.
\Cref{t:b} implies that
$$2^{d+1}-1 = \vc^\star(\C) 
\leq \vc^\star(\C')
\leq 2 \vc(\C')+1. $$
The left bound in \Cref{t:b} follows from \Cref{t:assouad}; it is tight for all VC dimensions as demonstrated by cubes (see Example~\ref{ex:cube} in \Cref{sec:examples}).
We leave it as an open question whether the upper bound in \Cref{t:b} is tight for all VC dimensions. In \Cref{sec:examples} we demonstrate that it is tight for \(d=1\) (\Cref{ex:d=1}) and 
notice that hyperplane arrangements yield extremal classes (in fact, maximum classes) with \(\mathtt{vc}(\mathcal{C})=d\) and \(\mathtt{vc}^\star(\mathcal{C}) = d+1\) (\Cref{ex:hyperplanes}).
\begin{open}\label{eq:dvc}
Is it true that for every \(d \in \mathbb{N}\) there exists an extremal class \(\mathcal{C}\) with \(\mathtt{vc}(\mathcal{C})=d\) and \(\mathtt{vc}^\star(\mathcal{C})=2d+1\)? Is there such a maximum class?
\end{open}

\subsection{Abstract Convexity}
To state our most general result we use the language of abstract convexity theory.
An (abstract) convexity space offers a simplified yet profound abstraction of Euclidean convexity. This concept originated in the work of \citet{levi1951helly} and was later given its current formalization by \citet*{kay1971axiomatic}. For a detailed introduction to this subject, readers may refer to the survey by \citet*{danzer1963helly}, or the more recent book \citet{van1993theory}.

A {\em convexity space} is a pair $(\Omega,\F)$ where $\F \subseteq 2^\Omega$ is a family of subsets that satisfies the following conditions:
\begin{itemize}
\item $\emptyset,\Omega\in \F$.
\item $\F$ is closed under intersections.
\end{itemize}
Sets within $\F$ are termed \emph{convex sets}. Convexity spaces are prevalent in mathematics, manifesting in forms such as closed sets in topological spaces, subgroups of groups, subtrees of graphs, and more. We note that in \citet{van1993theory}, the convexity space is also required to be closed under nested unions. However, all the convexity spaces that we use in this paper can be assumed to be finite, in which case this requirement becomes redundant.

Many basic concepts associated with convexity are still definable in this abstract setting. For instance, a convex set \( F \in \F \) is termed a \emph{half-space} if its complement is also convex.\footnote{In the standard Euclidean context, not all half-spaces are open or closed.} 
The \emph{convex hull} of a set \( P \subseteq \Omega \), denoted by \( \conv(P) = \conv_\F(P) \), is the intersection of all convex sets \( F \in \F \) that contain \( P \). These definitions enable the abstraction of classical results, such as Radon's Theorem~\citep{Radon:21}, Helly's Theorem~\citep{Helly:23}, and the Weak Epsilon-Net Theorem~\citep*{barany:1990,alon:1992,MoranY:20,holmsen:2021}.


\begin{definition}[Radon Numbers]
Let $(\Omega,\F)$ be a convexity space. We say that $p_1,\ldots, p_n\in \Omega$ are \emph{Radon-independent} if $\conv(\{p_i : i\in I\})\cap \conv(\{p_j : j\in J\})=\emptyset$ for every non-trivial partition $I\cup J = [n]$. The \emph{Radon number} of $\F$ is the largest size of a Radon-independent set.
\end{definition}

The classical Radon Theorem asserts that the Radon number of the convex sets in $\mathbb{R}^d$ is equal to $d+1$. 
We now explain how to 
associate a convexity space 
to every concept class.
These convexity spaces were implicitly studied in learning theory (more details follow).
We note that there are also other ways to associate a convexity space to a concept class
(for example, version spaces\footnote{Given a concept class $\C$ and a sequence of labeled examples \(\{(x_i,y_i)\}_{i=1}^m\), the associated version space is $\{c\in \C: c(x_i)=y_i \text{ for all }i\}$.} are convex). 

\begin{definition}[Convexity Space of a Class]\label{def-conv-space}
The convexity space of a class $\C\subseteq\{0,1\}^n$ is defined as follows.
The domain is $\Omega = \C$.
The half-spaces are sets of the form 
$$\C_{x,y}=\{c\in \C: c(x)=y\}$$ where $x\in [n]$ and $y\in\{0,1\}$.
The convex sets in $\F = \F_\C$ are arbitrary intersections of these half-spaces. 
\end{definition}

For $X\subseteq [n]$ and $t\colon X\rightarrow \{0, 1\}$, let us define 
$$\C_{X, t} = \{c\in \C: c(x) = t(x)\text{ for }x\in X\}.$$
Because $\C_{X, t} = \bigcap_{x\in X}\C_{x, t(x)}$, the set $\C_{X, t}$ is convex. The other way around, every non-empty convex set has this form. 


\begin{definition}[Radon Number of a Class]
The \emph{Radon number of a concept class $\C$}, denoted $\radon(\C)$, is the Radon number of the associated convexity space $\F_\C$.    
\end{definition}
We relate the dual Radon number to the VC dimension of extremal classes.

\begin{framed}
\vspace{-1em} 
\begin{customthm}{C}[Main Result III]\label{t:c}
 Let $\C \subseteq \{0,1\}^n$ be an extremal concept class. Then, 
\[\radon(\C) \leq 2\mathtt{vc}(\C)+1.\]
Moreover, the following lower bounds on $\radon(\C)$ hold
\begin{align*}
    \lfloor\log(2\vc(\C)+2)\rfloor&\leq \radon(\C) && \text{for any class $\C$}; \\
    \vc(\C)+1&\leq \radon(\C) && \text{for any maximum class $\C\neq \{0,1\}^n$}; \\
    \lfloor\log(2\vc(\C)+2)\rfloor&= \radon(\C) && \text{for $\C=\{0,1\}^n$}.
\end{align*}
\end{customthm}
\vspace{-1em} 
\end{framed}
The right bound in \Cref{t:b} follows from the first bound in \Cref{t:c}, 
because every $c_1,\ldots, c_k\in \C$ that is shattered by $\C^\star$ is Radon independent; in particular, $\vc^*(\C)\leq \radon(\C)$. Indeed if $I\cup J=[k]$ is a non-trivial partition then there exists $x$ such that $c_i(x)=1$ for all $i\in I$ and $c_j(x)=0$ for all $j\in J$. The complementing pair of half-spaces $\C_{x,1}$ and $\C_{x,0}$ separate $\{c_i: i\in I\}$ and $\{c_j: j\in J\}$ and hence their convex hulls are disjoint.

The upper bound on $\radon(\C)$ in \Cref{t:c} is the main technical result of the paper and is proven in \Cref{sec:overview}. \Cref{ex:hyperplanes} in \Cref{sec:examples} shows that it is also tight up to a factor of two. The two lower bounds on $\radon(\C)$ are proven in Propositions~\ref{p:t-c-ub1} and ~\ref{p:t-c-ub2} in \Cref{sec:proof}. Additionally, Examples~\ref{ex:cube} (cube)  and~\ref{ex:ball} (dented cube) in \Cref{sec:examples} show that these lower bounds are tight; the proofs that the values in these examples are as claimed are in Propositions~\ref{p:ex-cube} and~\ref{p:ex-ball} in   \Cref{sec:proof}.  
These lower bounds demonstrate that the dual Radon numbers can abruptly drop when adding a single concept.
Determining the tightness of the upper bound is left as an open question for future research.
\begin{open}
    Is it the case that for every \(d \in \mathbb{N}\) there exists an extremal class \(\mathcal{C}\) with \(\mathtt{vc}(\mathcal{C})=d\) and \(\radon(\mathcal{C})=2d+1\)? Is there such a maximum class?
\end{open}

Our final result summarizes the relationship between the dual Radon number and dual VC dimension.
This result, while less central, completes the triangular relationship between VC, dual VC, and dual Radon. 

\begin{framed}
\begin{customthm}{D}\label{t:d}
\begin{enumerate}
    \item 
Every concept class $\C$ satisfies $\vc^\star(\C)\leq \radon(\C)$. 
\item For general classes, $\radon(\C)$ is not upper-bounded by $\vc(\C)$ or $\vc^*(\C)$. 
\item If $\C$ is extremal, then $\radon(\C) \leq 2^{\vc^\star(\C)+2}-1$, and this bound is tight up to a multiplicative factor.
\end{enumerate}
\end{customthm}
\end{framed}

\Cref{t:d} is proven in Proposition~\ref{p:t-d} in \Cref{sec:proof}.
The fact that the dual Radon number of a class is not upper-bounded by its VC and dual VC dimensions is demonstrated by the class of singletons (\Cref{ex:singletons} in \Cref{sec:examples}). \Cref{ex:ball} (dented cube) in \Cref{sec:examples} demonstrates that for extremal classes, the dual Radon number can indeed be exponential with respect to the dual VC dimension.

\vspace{1mm}

\paragraph{Organization.}
The rest of this manuscript is organized as follows.
Section~\ref{sec:overview} provides a proof of our main result, that is, $\radon(\C) \leq 2\vc(\C) + 1$ for extremal class $\C$; it also introduces the relevant topological notions and theorems. 
In \Cref{sec:examples} we present several, rather simple, examples that demonstrate the sharpness of our bounds. The supplementary proofs, all of which are mostly elementary, are provided in \Cref{sec:proof}; this includes showing that the values of the relevant parameters in examples are indeed what they are claimed to be. Finally, in \Cref{sec:proof-discussion} we give a broader discussion on the topological notions introduced in \Cref{sec:overview}.

\section{Main Proof}\label{sec:overview}

We now present the proof for the upper bound in \Cref{t:c}. As previously discussed, \Cref{t:a} and \Cref{t:b} are derived as corollaries.
Consider an extremal class $\C \subseteq \{0,1\}^n$; our goal is to demonstrate that $\radon(\C) \leq 2\vc(\C) + 1$. 
The proof employs a topological approach. 
For this purpose, we define a cube complex (see Figure~\ref{fig:completion}).
We think of the cube complex as a topological subspace of $[0,1]^n$.
It comprises cubes of various dimensions that are glued together according to the structure of $\C$.

\begin{definition}[The Cube Complex]
The \emph{cube complex} $Q = Q(\C)\subseteq [0,1]^n$ of $\C \subseteq \{0,1\}^n$ is defined as follows.   
The vertices of $Q$ are the concepts of $\C$.
The cubes of $Q$ are the ``filled cubes of $\C$''. 

Formally, a $d$-cube in $\C$ is a pair $(Y, f)$, where $Y \subseteq [n]$ is of size $|Y|=d$ and $f$ is a function $f\colon [n] \setminus Y\rightarrow \{0,1\}$ such that for every $g\colon Y\rightarrow \{0,1\}$ there is a concept $c\in \C$ such that $c(x) = g(x)$ for every $x\in Y$ and $c(x) = f(x)$ for every $x\in [n] \setminus Y$.

For a cube $(Y, f)$ of $\C$, we define the corresponding \emph{solid cube} $(Y, f)_Q\subseteq [0,1]^n$ as 
$$(Y, f)_Q = \{v\in [0,1]^n: v(x) = f(x)\text{ for }x\in Y\}.$$
We now define the cubes of $Q$ as the collection $(Y, f)_Q$, over all cubes $(Y, f)$ of $\C$.
\end{definition} 
Let us note several basic facts about the cube complex $Q(\C)$ that we are going to utilize:
\begin{itemize}
	\item For a $d$-cube $(Y, f)$ of $\C$, $d= |Y|$ is called the \emph{dimension} of this cube, and it coincides with the topological dimension of $(Y, f)_Q$. We define the \emph{dimension} of $Q = Q(\C)$, denoted $\dim(Q)$, as the maximum dimension of its cube;
	\item $0$-cubes of $V$, also called \emph{vertices} of $V$, are precisely the concepts of $\C$. Here we assume that a concept $c\in \C\subseteq \{0, 1\}^n$ is naturally identified with the point in $[0,1]^n$ with the same coordinates. Moreover, every point of $Q(\C)$ with all coordinates in $\{0,1\}$ corresponds to a concept of $\C$. This essentially means that $\C$ is ``embedded'' into $Q(\C)$;
	\item We say that $(Y_1, f_1)$ is a \emph{subcube} of $(Y_2, f_2)$ if $Y_1\subseteq Y_2$ and $f_2$ is a restriction of $f_1$ to $[n] - Y_2$. Then $(Y_1, f_1)_Q \subseteq (Y_2, f_2)_Q$ if and only if $(Y_1, f_1)$ is a subcube of $(Y_2, f_2)$. In this case we also say that $(Y_1, f_1)_Q$ is a subcube of $(Y_2, f_2)_Q$;
	\item Any subcube of a cube $(Y, f)_Q$ of $Q(\C)$ is a cube of $Q(\C)$. Also, any cubes $(Y_1, f_1)$ and $(Y_2, f_2)$ of $Q(\C)$ either disjoint, or intersect by a cube of $Q(\C)$. This property enables us to indeed call $Q(\C)$ a cube \emph{complex}, as opposed to it being just a collection of cubes. 
\end{itemize}
\begin{figure}[hbt]
\centering
\begin{tikzpicture}

\begin{scope}[scale=2, xshift=-1.5cm]  
    \coordinate (O) at (0,0,0);
    \coordinate (A) at (0,1,0);
    \coordinate (B) at (1,0,0);
    \coordinate (C) at (1,1,0);
    \coordinate (D) at (0,0,1);

    \draw[fill=black] (O) circle (1pt) node[anchor=east] {(0,0,0)};
    \draw[fill=black] (A) circle (1pt) node[anchor=east] {(0,1,0)};
    \draw[fill=black] (B) circle (1pt) node[anchor=west] {(1,0,0)};
    \draw[fill=black] (C) circle (1pt) node[anchor=west] {(1,1,0)};
    \draw[fill=black] (D) circle (1pt) node[anchor=east] {(0,0,1)};

    \draw[dotted] (O) -- (A) -- (C) -- (B) -- cycle;
    \draw[dotted] (O) -- (D);
\end{scope}

\begin{scope}[scale=2, xshift=1.5cm]  
    \coordinate (O) at (0,0,0);
    \coordinate (A) at (0,1,0);
    \coordinate (B) at (1,0,0);
    \coordinate (C) at (1,1,0);
    \coordinate (D) at (0,0,1);

    \draw[fill=black] (O) circle (1pt) node[anchor=east] {(0,0,0)};
    \draw[fill=black] (A) circle (1pt) node[anchor=east] {(0,1,0)};
    \draw[fill=black] (B) circle (1pt) node[anchor=west] {(1,0,0)};
    \draw[fill=black] (C) circle (1pt) node[anchor=west] {(1,1,0)};
    \draw[fill=black] (D) circle (1pt) node[anchor=east] {(0,0,1)};

    \draw[fill=gray] (O) -- (A) -- (C) -- (B) -- cycle;
    \draw[fill=gray] (O) -- (D);
\end{scope}

\draw[->, thick] (0.5,1) -- (1.5,1);

\end{tikzpicture}
\caption{\small A $2$-dimensional illustration of the cube complex for the extremal class $\C = \{000, 010, 110, 100, 001\}$. It has $5$ vertices ($0$-dimensional cubes), $5$ edges ($1$-dimensional cubes), and $1$ square ($2$-dimensional cube). The square corresponds to a cube $(Y, f)$ of $\C$ with $Y=\{1,2\}$ and $f\colon 3\mapsto 0$. A unique maximal edge, connecting $(0,0,0)$ and $(0,0,1)$, corresponds to a cube with $Y=\{3\}$ and $f\colon 1\mapsto 0, 2\mapsto 0$.}\label{fig:completion}
\end{figure}
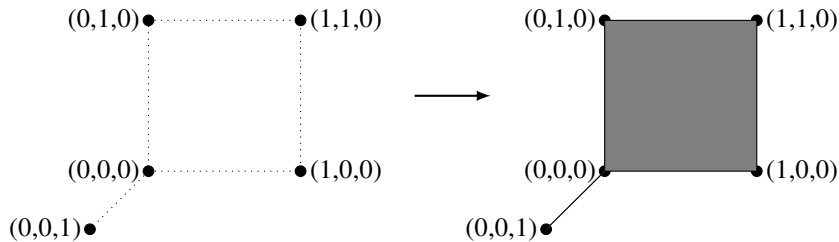

For $Y\subseteq [n]$, we say that $Y$ is \emph{strongly shattered} by $\C$ if there is a cube $(Y, f)$ in $\C$. It is well-known that for extremal classes strong shattering is equivalent to shattering; see, for example, \cite{Bollobas:95}.\footnote{The authors use the terms \emph{strongly traced} and \emph{traced}.} It then follows that for an extremal $\C$, $\dim(Q) = \vc(\C)$. 

A key concept that we utilize in the proof is the \emph{Radon number} of a topological space.
Denote by $\Delta^{d}$ the standard $d$-dimensional simplex (with $d+1$ vertices), and by $\partial \Delta^{d}$ its boundary, which is homeomorphic to the $(d-1)$-dimensional sphere $S^{d-1} \subset \R^d$.

\begin{definition}[Radon number of a topological space]
	For a topological space $X$, the \emph{Radon number} $\radon(X)$ is the largest integer $d$ such that there is a continuous mapping $f\colon \partial \Delta^{d+1} \rightarrow X$ such that the images of disjoint faces are disjoint.
The map $f$ is called a \emph{Radon map}.
\end{definition}

Alternatively, $\radon(X)+1$ is the minimal integer $d$ such that for every continuous map $f\colon \partial \Delta^{d+1} \rightarrow X$, there is a pair of disjoint faces whose images intersect. The name ``Radon number'' is motivated by the famous Topological Radon Theorem. 

\begin{theorem}[Topological Radon Theorem~\citep*{Bajmoczy:79}]\label{t:toprad}
$$\radon(\R^d) = d-1.$$
\end{theorem}

Our proof consists of two main steps, whose statements are now easy to formulate. First, we prove that for an extremal class $\C$ the ``topological'' Radon number $\radon(Q(\C))$ can be lower-bounded by the ``discrete'' Radon number $\radon(\C)$.
\begin{proposition}\label{prop:embed}
    Let $\C$ be an extremal class, and let $c_1, \ldots, c_k \in \C$ be Radon independent. Then, there exists a Radon map $f\colon \partial \Delta^{k-1} \rightarrow Q(\C)$. 
    In particular, $\radon(Q(\C)) \geq \radon(\C) - 2$.
\end{proposition}

The second step upper bounds the Radon number of an arbitrary simplicial complex in terms of its dimension.

\begin{lemma}\label{lem:radon-for-sc}
	For a finite $d$-dimensional simplicial complex $K$, we have $d-1\leq \radon(K)\leq 2d-1$. 
\end{lemma}
The slightly weaker upper bound $\radon(K) \leq 2d$ can be derived as a consequence of the Topological Radon Theorem. Indeed, by the geometric realization theorem, $K$ can be embedded in $\R^{2d+1}$. Thus, any Radon map into $K$ 
yields a Radon map into $\R^{2d+1}$, and so $\radon(K) \leq \radon(\R^{2d+1}) = 2d$.

The required bound can be obtained by combining Proposition~\ref{prop:embed} and Lemma~\ref{lem:radon-for-sc} as follows. Let $Q^t(\C)$ be a \emph{triangulation} of $Q$. That is, $Q^t$ is a finite simplicial complex such a) every simplex of $Q^t$ is contained in a cube of $Q$ of the same dimension, and b) every cube of $Q$ is a finite union of simplices of $Q^t$ of the same dimension. In particular, $\dim(Q) = \dim(Q^t)$ and $Q$ and $Q^t$ are homeomorphic as topological spaces, which implies $\radon(Q) = \radon(Q^t)$. It is well-known and easy to check that some triangulation $Q^t$ of $Q$ exists, see for example \cite*{triangulations}. Then
\begin{align*}
	\radon(\C) &\leq \radon(Q(\C)) + 2 = \radon(Q^t(\C)) + 2 
		\\&\leq 2\dim(Q^t(\C)) - 1 + 2 = 2\dim(Q(\C)) - 1 + 2= 2 \vc(\C) + 1.
\end{align*}

To complete the proof, we thus need to prove Proposition~\ref{prop:embed} and Lemma~\ref{lem:radon-for-sc}. For the proof of Proposition~\ref{prop:embed}, let us first the following property of $Q(\C)$.
Recall that a topological space $Y$ is $k$-connected if for every $0\leq \ell \leq k$, any continuous map from an $\ell$-dimensional sphere $S^\ell$ into $Y$ can be extended to a continuous map from the $(\ell+1)$-dimensional ball $B^{\ell+1}$ into $Y$.
For $k=1$, for example,
this means that every copy of a circle in $Y$ can be extended to a copy of a disc in $Y$.
Figuratively, we can freely deflate balloons inside $Y$.

\begin{proposition}\label{prop:conv}
	Let $\C$ be an extremal class, and let $Y\subseteq \C$ be an arbitrary non-empty convex, with respect to the convexity space $\F_\C$, subset of $\C$. Then $Q(Y)\subseteq Q(\C)$ is $k$-connected for every $k\geq 0$.
\end{proposition}
\begin{proof}
The proof relies on the following two results from \cite*{Chalopin:22} that emphasize the highly connected nature of extremal classes (the authors use the term \emph{ample} for extremal):
\begin{itemize}
	\item If $\C$ is extremal then every $\C_{X, t}$ is extremal.
	\item If $\C$ is extremal then the cubical complex $Q(\C)$ is \emph{collapsible}.
\end{itemize}
In \cite{Chalopin:22}, the first bullet is Theorem~3.1~(6), but it is also a well-known property of extremal classes. The second is their Proposition~4.12, and, to our knowledge, it was not known before that. We refer the reader to the referenced paper for the definition of collapsibility; here we only use the standard fact\footnote{\cite{Chalopin:22}, e.g., say that 
collapsibility is a stronger version of contractibility.
It is well-known that a contraction of $Y$ enables to extend any continuous map  $S^\ell \cong \partial B^{\ell+1} \to Y$ to a continuous map $B^{\ell+1} \to Y$. Roughly speaking, the extension is constructed by emulating inside $Y$ the contraction of $S^\ell$ to a point in $B^{\ell+1}$ using the contraction of $Y$.} that collapsible spaces are $k$-connected for all $k$.  

Now let us take a convex set $Y\subseteq \C$. By Definition~\ref{def-conv-space} and the discussion afterwards, $Y = \C_{X, t}$ for some $X\subseteq [n]$ and $t\colon X\rightarrow \{0,1\}$. Then by the first item above, $Y$ itself is extremal. And by the second item, $Q(Y)$ is collapsible, and hence $k$-connected for all $k$.
\end{proof}

\begin{proof}(Of Proposition~\ref{prop:embed}). Let $c_1, \ldots, c_k \in \C$ be Radon independent,
and let $\Delta = \Delta^{k-1}$.
Recall that the abstract simplices of $\partial\Delta$ are precisely non-empty subsets of $[k]$.
We are going to define $f\colon \partial\Delta \rightarrow Q = Q(\C)$  such that for every simplex $s$ with vertex set $V(s)$ in $\partial\Delta^{k-1}$, 
$$f(s) \subseteq \conv_{Q} (C(s)),$$ where $C(s) = \{c_i: i\in V(s)\}$ and $\conv_Q(W) = Q(\conv(W)) \subseteq Q(\C)$, for $W\subseteq \C$.

The map $f$ is defined inductively. 
The $\ell$-skeleton of $\partial \Delta$ comprises simplices of dimension at most $\ell$.
We construct $f$ first on the $0$-skeleton, then on the $1$-skeleton, and so forth. 
For $\ell=0$, we need to define $f$ on $V(\partial \Delta) = [k]$. Naturally, we do it by putting $f(i) = c_i$, for $i\in [k]$. Note that here we use the fact that $\C$ is embedded into $Q(\C)$, that is, each $c_i\in \C \subseteq \{0,1\}^n$ is treated as a point in $Q(\C)\subseteq [0,1]^n$. Then trivially, for every $i\in [k]$,  $f(i)\in \conv_{Q} (c_i)$, as needed.

Now, assuming that $f$ is defined on the $\ell$-skeleton and we need to define it on the $(\ell+1)$-skeleton.  By the induction hypothesis, for every maximal $(\ell+1)$-dimensional simplex $s$, the map $f$ is already defined on its $\ell$-skeleton $\partial s$. Moreover, for each $\ell$-dimensional simplex $s'$ in $s$, $$f(s')\subseteq \conv_{Q} (C(s')) \subseteq \conv_{Q} (C(s)).$$ As $s\cong B^{\ell+1}$ and $\partial s \cong S^\ell$, applying Proposition~\ref{prop:conv} to $s$ with $Y = \conv (C(s))$ gives an extension of $f$ to $s$ with $f(s)\subseteq \conv_{Q} (C(s))$.

This concludes the construction of the map $f$. Let us take any disjoint simplices $s_1$ and $s_2$ of $\partial\Delta^{k-1}$. In particular, $V(s_1)$ and $ V(s_2)$ are disjoint. As $c_1, \dots, c_k$ are Radon independent, the convex sets $\conv(C(s_1))$ and $\conv(C(s_2))$ are disjoint. 

We claim that then $Q_1 = \conv_{Q}(C(s_1))$ and $Q_2 = \conv_{Q}(C(s_2))$ are also disjoint. Indeed, both $Q_1$ and $Q_2$ are cube subcomplexes of $Q(\C)$, and hence their intersection $Q_1\cap Q_2$ is also a cube subcomplex of $Q(\C)$. Hence, if $Q_1\cap Q_2$ is nonempty, it contains a $0$-dimensional cube, that is, a vertex, $p$. But vertices of a cube complex of a class are in bijective correspondence with the elements of this class, that is, $p\in Q_1 =  Q(\conv(C(s_1)))$ if and only if $p\in \conv(C(s_1))$, and similarly, $p\in \conv(C(s_1))$, a contradiction.

Thus, the images $f(s_1)$ and $f(s_2)$ are disjoint for any disjoint simplices $s_1$ and $s_2$ of $\partial\Delta^{k-1}$, and hence $f$ is Radon, as needed.
\end{proof}

Finally, let us prove Lemma~\ref{lem:radon-for-sc}. 
Here we rely on two known results. 
The first is proved by \cite*{guilbault:10} in his simple proof of the Topological Radon Theorem using the Borsuk-Ulam Theorem. 
\begin{proposition}[\cite{guilbault:10}, Proposition 3.1]\label{prop:guilbault}
	For every $k\geq 0$, there is a continuous function $\lambda_k\colon S^k \rightarrow \partial \Delta^{k+1}$ such that images of the antipodal points lie in disjoint simplices.	
\end{proposition}

The second result is a variant of the Borsuk-Ulan Theorem for general simplicial complexes
that was proven by \cite*{jaw:02}.

\begin{theorem} [\cite{jaw:02}] \label{thm:Jaw}
Let $K$ be a finite simplicial complex of dimension $d$. If $f \colon S^k \to K^d$ is continuous and $2d \leq k$, then there exists $x \in S^k$ so that $f(x) = f(-x)$. 
\end{theorem}

\begin{proof}(Of Lemma~\ref{lem:radon-for-sc}).
The lower bound $d-1\leq \radon(K)$ holds because $K$ contains a $d$-dimensional simplex. 
For the upper bound $\radon(K)\leq 2d-1$, 
we use the two results above.
Given a Radon map 
$f\colon \partial \Delta^{k+1} \rightarrow K$, we can use the map $\lambda_k$ to get a \emph{Borsuk-Ulam map} $g = f\circ \lambda_k$. That is, the map $g\colon S^k\rightarrow \R^d$ does not collapse antipodal points. 
By \Cref{thm:Jaw},
we can deduce that
$k \leq 2d-1$.
\end{proof}


\section{Examples}\label{sec:examples}

We now give examples that demonstrate that our bounds relating the Radon number, dual VC dimension, and VC dimension are nearly tight. These examples reveal that the lower bounds in Theorems~\ref{t:b}, \ref{t:c}, and \ref{t:d} are tight, while the upper bounds are tight within constant multiplicative factors. The determination of more precise bounds is left open for future research. This is summarized in \Cref{tab:bounds} below. The detailed analysis of the examples appears in \Cref{sec:proof}.

\begin{table}[ht]
\centering
\begin{tabular}{|l|c|c|c|}
\hline
             & \begin{tabular}[c]{@{}c@{}}\Cref{t:b} \\ (VC vs.\ dual VC)\end{tabular} & \begin{tabular}[c]{@{}c@{}}\Cref{t:c} \\ (VC vs.\ Radon)\end{tabular} & \begin{tabular}[c]{@{}c@{}}\Cref{t:d} \\ (dual VC vs.\ Radon)\end{tabular} \\ \hline
Upper Bound  
    & {\small Examples~\ref{ex:d=1} and~\ref{ex:hyperplanes}} 
    & {\small Examples~\ref{ex:ball}, \ref{ex:d=1}, and~\ref{ex:hyperplanes}}  
    & \small{Example~\ref{ex:ball}}          
\\ \hline
Lower Bound  
    & {\small Example~\ref{ex:cube}} 
    & {\small Examples~\ref{ex:cube}, \ref{ex:ball}, and~\ref{ex:hyperplanes}} 
    & {\small Examples~\ref{ex:cube} and~\ref{ex:hyperplanes}}  \\ \hline
\end{tabular}
\caption{Summary of examples demonstrating the tightness of our bounds. Our lower bounds are tight, while the upper bounds are tight within constant multiplicative factors.}
\label{tab:bounds}
\end{table}

Our initial example demonstrates that for non-extremal concept classes $\C$, the Radon number cannot be upper-bounded solely as a function of the VC or dual VC dimensions. (However, the lower bound by the dual VC dimension remains valid even for non-extremal classes, as stated in \Cref{t:d}.)

\begin{example}[Singletons]\label{ex:singletons}
   The class of singletons \(\C=\{1_{\{i\}} : i\in [n]\}\)
   is not extremal, and
    \begin{enumerate}
        \item \(\vc(\C)=1\).
        \item \(\vc^\star(\C)=1\),
        \item \(\radon(\C)=n\).
    \end{enumerate}
    This illustrates that for non-extremal classes, the Radon number can be unbounded while both the VC and dual VC dimensions remain as low as $1$.
\end{example}

The upcoming examples pertain to Theorems~\ref{t:b}, \ref{t:c}, and~\ref{t:d}. Refer to Table~\ref{tab:bounds} for a roadmap linking each example to its theorem.

\begin{example}[Cube]\label{ex:cube}
    The cube $\C=\{0,1\}^{d}$ is a maximum class satisfying:
    \begin{enumerate}
        \item $\vc(\C)=d$,
        \item $\vc^\star(\C)=\lfloor \log d \rfloor$,
        \item $\radon(\C)= \lfloor \log(2d+2)\rfloor$.
    \end{enumerate}
    This demonstrates the tightness of the first lower bound in \Cref{t:c} and the lower bound in \Cref{t:b}.
\end{example}

\begin{example}[Dented cube]\label{ex:ball}
The dented cube $\C$
is obtained by removing the all-$1$ function from the cube $\{0,1\}^{d+1}$.
It is a maximum class satisfying:
\begin{enumerate}
    \item \(\vc(\C)=d\),
    \item $\vc^\star(\C)=\lfloor\log(d+1)\rfloor$,
    \item $\radon(\C)=d+1$.
\end{enumerate}
This example demonstrates that the second lower bound in \Cref{t:c} is tight. Additionally, it shows that the upper bounds in \Cref{t:c} and \Cref{t:d} are tight within a constant multiplicative factor.
\end{example}

\begin{example}\label{ex:d=1}
Let 
\begin{align*}
\C  = \{01010101, & \allowbreak 11010101, \allowbreak \bb{10010101}, 
 		\allowbreak 01110101, \allowbreak \\ & \bb{01100101}, 
 		\allowbreak 01011101, \allowbreak \bb{01011001}, 
 		\allowbreak 01010111, \allowbreak \bb{01010110}\}.
\end{align*}
This class is maximum and satisfies:
\begin{enumerate}
    \item $\vc(\C)=1$,
    \item $\vc^\star(\C)=3$ (any three out of the four blue concepts are dually shattered),
    \item $\radon(\C)=3$.
\end{enumerate}
This confirms that the upper bounds in \Cref{t:b} and \Cref{t:c} are tight for $d=1$.
\end{example}

\subsection{Hyperplane Arrangements}
We conclude with a geometric example of linear classifiers in $\mathbb{R}^d$.
Consider a collection of hyperplanes $H_1,\ldots, H_n\subseteq \mathbb{R}^d$
where each $H_i$ is defined as $H_i = \{x : \langle a_i, x \rangle = b_i\}$ with $a_i\in\mathbb{R}^d$ and $b_i\in\mathbb{R}$.
For each point $x\in\mathbb{R}^d\setminus(\cup_{i=1}^n H_i)$, associate a sign vector $c_x\in\{\pm \}^n$ defined by:
\[c_x(i) = 
\begin{cases}
+ & \text{if } \langle a_i, x \rangle > b_i,\\
- & \text{if } \langle a_i, x \rangle < b_i.
\end{cases}\]
Define the class $\C=\{c_x : x\in\mathbb{R}^d\}$. 
Each $c\in\C$ corresponds to a cell induced by the hyperplane arrangement (see Figure~\ref{fig:hyperplane} for an illustration).
We say that the hyperplanes are generic if
for every subset of hyperplanes $\{H_{i_1},\ldots, H_{i_k}\}$, the intersection $\cap_{j=1}^kH_{i_j}$ has co-dimension $k$
(and every $d+1$ hyperplanes
do not intersect).
For every $d,n$,
there are generic hyperplane arrangements, and as the name suggests a ``typical'' arrangement is generic. 
If the hyperplanes are generic, the corresponding $\C$ is a maximum class of VC dimension $d$~\citep{Moran:12}.

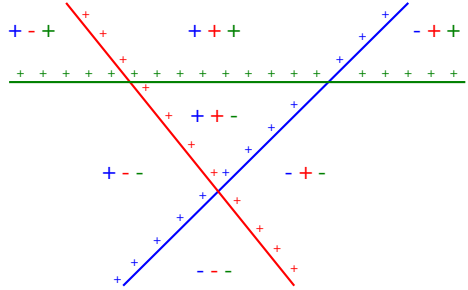
\begin{figure}[hbt]
\centering
\begin{tikzpicture}[scale=1.5]

\draw[color=blue, thick] (-0.5, -0.5) -- (2, 2); 
\draw[color=red, thick] (-1, 2) -- (1, -0.5); 
\draw[color=darkgreen, thick] (-1.5, 1.3) -- (2.5, 1.3); 

\foreach \x in {-0.55,-0.4,-0.2,0,0.2,0.4,0.6,0.8,1,1.4,1.6,1.8} {
    \node[color=blue] at (\x, \x+0.1) {\tiny +};
}

\foreach \x in {-1.4,-1.2,-1.0,-0.8,-0.6,-0.4,-0.2,0,0.2,0.4,0.6,0.8,1,1.2,1.6,1.8,2.0,2.2,2.4} {
    \node[color=darkgreen] at (\x, 1.375) {\tiny +};
}

\foreach \point in {(-0.825,1.9), (-0.675, 1.725), (-0.5, 1.5), (-0.3, 1.25), (-0.1, 1.0), (0.1, 0.75), (0.3, 0.5), (0.5, 0.25), (0.7, 0), (0.85, -0.175), (1,-0.35)} {
    \node[color=red] at \point {\tiny +};
}

\node at (-0.5, 0.5) {\small \color{blue}+ \color{red}- \color{darkgreen}-};
\node at (0.3, -0.4) {\small \color{blue}- \color{red}- \color{darkgreen}-};
\node at (1.1, 0.5) {\small \color{blue}- \color{red}+ \color{darkgreen}-};
\node at (0.3, 1) {\small \color{blue}+ \color{red}+ \color{darkgreen}-};
\node at (0.3, 1.75) {\small \color{blue}+ \color{red}+ \color{darkgreen}+};
\node at (2.25, 1.75) {\small \color{blue}- \color{red}+ \color{darkgreen}+};
\node at (-1.3, 1.75) {\small \color{blue}+ \color{red}- \color{darkgreen}+};

\end{tikzpicture}
\caption{\small A $2$-dimensional hyperplane arrangement consisting of $3$ lines and $7$ cells. The class $\C$ 
is maximum.
Its VC dimension is two because $|\C|<8$.
Its dual VC dimension is one 
because $|\C^\star|<4$.
Its Radon number is three,
because the three concept
$+-+,-++,---$ are Radon indepedent.
}\label{fig:hyperplane}
\end{figure}

The convexity space $\F = \F_C$
corresponds to the standard notion of convexity. 
The half-space
$\C_{i,y}$ corresponds to the positive or negative, depending on $y$, half-space defined by $H_i$ in $\R^d$.
The classical Radon Theorem
implies that $\radon(\C) \leq d+1$ for all hyperplane arrangements in $\R^d$ (not necessarily generic). 
\Cref{t:d} says that
$\vc^\star(\C)\leq \radon(\C) \leq d+1$.

Take $d+1$ generic hyperplanes that support a simplex. 
The total number of cells we get is $2^{d+1}-1$.
The VC dimension is $d$.
The Radon number is $d+1$ because the cells obtained by flipping (with respect to the simplex) all hyperplanes but one are Radon independent. 
The dual VC dimension 
is however at most $\log (d+1)$.
In fact, we obtained a different description of \Cref{ex:ball}.

To make the dual VC dimension larger, let $p_1,\ldots,p_{d+1} \in \R^d$ be $d+1$ points that are shattered by hyperplanes
$h_1,\ldots,h_{2^{d+1}}$.
These points can be chosen, for example, as the vertices of a full-dimensional simplex.
The hyperplanes can be chosen to be generic so that we get a maximum class of VC dimension $d$.
By construction, the dual VC dimension is at least $d+1$,
and, as discussed above,
it is also at most $d+1$.


\begin{example}\label{ex:hyperplanes}
    There are maximum classes $\C$ defined by arrangements of generic hyperplanes in $\mathbb{R}^d$ fulfilling:
    \begin{enumerate}
        \item $\vc(\C)=d$,
        \item $\vc^\star(\C)=d+1$,
        \item $\radon(\C)= d+1$.
    \end{enumerate}
    This example demonstrates the tightness of the lower bounds in \Cref{t:c} and \Cref{t:d}, and that the upper bounds in \Cref{t:b} and \Cref{t:c} are tight within a factor of $2$.
\end{example}

A more general type of classes, called \emph{pseudogeometric range spaces}, was defined and studied in~\cite{Gartner:94}.
Intuitively, this class is an abstraction of the class obtained from generic hyperplanes in $\R^d$, with pseudo-hyperplanes instead of hyperplanes, where pseudo-hyperplanes
are topological abstractions of linear hyperplanes. 
Their definition, however, is combinatorial and is inductively defined based
on the following two ideas.
First, a one-dimensional pseudogeometric range space is a, combinatorially defined, line.
Second, for $d\geq 2$, a class is a $d$-dimensional pseudogeometric range space if its intersection with any pseudo-hyperplane from it is a $(d-1)$-dimensional pseudogeometric range space; here again, the intersection is defined combinatorially. 

Our work suggests a different way to define and think of pseudo-hyperplane arrangements. 
Let us first imagine a generic hyperplane arrangement realized in $\R^d$,
and denote by $\C$ its concept class. The cube complex $Q = Q(\C)$ can be thought of as a dual object to the geometric realization.
Each point in $Q$ is a cell in $\C$. Two adjacent cells are connected by an edge in $Q$,
four cells that meet at a point define a square, and so forth.
Thus, $Q$, as a topological space, can be embedded into $\R^d$. 

%

Based on this observation, it is natural to suggest that a maximum class $\C$ of VC dimension $d$ is a $d$-dimensional pseudogeometric range space in the sense of~\cite{Gartner:94} if and only if its cube complex $Q(\C)$ is embeddable into $\R^d$. This equivalence will be the topic of a future study.

\section{Proofs and Analysis of Examples}\label{sec:proof}

Here we analyze all examples in~\Cref{sec:examples}.
Some of the bounds in our main results are deduced as corollaries.

\begin{proposition}[First Lower Bound in Theorem~\ref{t:c}]\label{p:t-c-ub1}
	For any class $\C$, $$\lfloor\log(2\vc(\C)+2)\rfloor\leq \radon(\C).$$
\end{proposition}
\begin{proof}
	Suppose $\vc(\C)=d$, and, without losing generality, assume that $\C$ shatters $[d]$. Let $k=\lfloor\log(2d+2)\rfloor$. In particular, $D = (2^k - 2)/2 \leq d$. We say that $(I, J)$ is a (non-trivial, ordered) partition of $[k]$ if $I, J\neq\emptyset$, $I\sqcup J = [k]$, and $1\in I$. Let $\I$ be a set of all such partitions; the last condition ensures that precisely one of $(I, J)$ and $(J, I)$ is in $\I$. It is easy to see that then $|\I| = (2^k -2)/2$, and let us take an arbitrary bijection $(I, J)\mapsto x_{I, J}$ from $\I$ to $[D]$. 
	
	Let us now take $c_i\in \C$, for $i\in [k]$, such that for every $(I, J)\in \I$, $c_i(x_{I, J}) = 0$ if $i\in I$ and $1$ if $i\in J$. Such $c_i$ exist because $\C$ shatters $[D]\subseteq [d]$. We claim that the concepts $c_1, \dots, c_k$ are Radon independent. Indeed, let $(I, J) \in \I$ be a partition of $[k]$, and let $x=x_{I, J}$. Then $c_i(x) = 0$ for all $i\in I$, and $c_j(x) = 1$ for all $j\in J$. But then $I\subseteq \C_{x, 0}$ and $J\subseteq \C_{x, 1}$, which implies that $\conv_\C(I)\cap \conv_\C(J) \subseteq \C_{x, 0} \cap \C_{x, 1} = \emptyset$, as needed. 
\end{proof}

\begin{proposition}[Example~\ref{ex:cube} - Cube]\label{p:ex-cube}
    For the cube $\C=\{0,1\}^{d}$:
    \begin{enumerate}
        \item $\vc(\C)=d$,
        \item $\vc^\star(\C)=\lfloor \log d \rfloor$,
        \item $\radon(\C)= \lfloor \log(2d+2)\rfloor$.
    \end{enumerate}
\end{proposition}
This example also confirms the last lower bound in Theorem~\ref{t:c}.
\begin{proof}
 The equality $\vc(\C)=d$ is trivial. The bound $\vc^\star(\C)=\lfloor \log d \rfloor$ is also trivial and known from the sharpness of the bound in Theorem~\ref{t:assouad}; see \cite{assouad:83}. 
From Proposition~\ref{p:t-c-ub1}, $\radon(\C)\geq \lfloor \log(2d+2)\rfloor$. It remains to prove that $\radon(\C) \leq \log(2d+2)$. The argument goes along the lines of the proof of Proposition~\ref{p:t-c-ub1}.

We first prove the followng technical claim: \emph{for $I, J\subseteq \C$, 
     $\conv_\C(I) \cap \conv_\C(J)=\emptyset$ if and only if there is $x\in [d]$ that \emph{separates} $I$ and $J$, that is, such that $c(x)=0$ for all $c\in I$ and $c(x)=1$ for all $c\in J$, or the other way round.}

	Indeed, if there is such $x$, then $\conv_\C(I)$ and $\conv_\C(J)$ are trivially disjoint. In the other direction, suppose $\conv_\C(I) \cap \conv_\C(J)=\emptyset$. As explained below Proposition~\ref{prop:conv}, every non-empty convex set $Y$ in $\F_\C$ is of the form $Y = \C_{X, t}$ for $X\subseteq [d]$ and $t\colon X\rightarrow \{0,1\}$. Write $\conv_\C(I) = \C_{X_I, t_I}$ and $\conv_\C(J) = \C_{X_J, t_J}$. As $\conv_\C(I)$ and $\conv_\C(J)$ are disjoint, there is $x\in X_I\cap X_J$ such that $t_I(x) = 0$ and $t_J(x) = 1$, or vice versa. But, by the definition of $\C_{X, t}$, this means that  $c(x)=0$ for all $c\in I$ and $c(x)=1$ for all $c\in J$, or the other way round, as needed.
	
	Now, suppose that $c_1, \dots, c_k \in \C$ are  Radon shattered. We see that for each non-trivial ordered partition $(I,J)$ of $[k]$, there is $x_{I, J} \in [d]$ that separates $\{c_i : i \in I\}$ and $\{c_j :  j \in J\}$. Moreover, it is easy to see that the map $(I, J)\mapsto x_{I, J}$ on $\I$ is one-to-one. As $|\I| = (2^k-2)/2$, it follows that $d \geq (2^k-2)/2$.
\end{proof}

\begin{proposition}[Lower Bound for Maximum Classes in Theorem~\ref{t:c}]\label{p:t-c-ub2}
For any maximum class $\C\neq \{0,1\}^n$, we have $\vc(\C)+1\leq \radon(\C)$.
\end{proposition}

We use the following known property of extremal classes (see \cite*{Chalopin:22} or \cite*{chornomaz:22}). For a class $\C\subseteq \{0,1\}^n$, a set $X\subseteq [n]$ is called \emph{minimal non-shattered} if $X$ is not shattered by $\C$ and $Y$ is shattered by $\C$ for every $Y\subsetneq X$. 

\begin{lemma}
    For every extremal class $\C$, if $X$ is minimal non-shattered by $\C$, then $X$ contains a unique forbidden trace $t$; that is, there is $t\colon X\rightarrow \{0, 1\}$ such that
\begin{itemize}
	\item For every $t'\colon X\rightarrow \{0, 1\}$ such that $t'\neq t$, there is $c\in \C$ such that $t'(x) = c(x)$ for all $x\in X$.
	\item There is no $c\in \C$ such that $t(x) = c(x)$ for all $x\in X$.
\end{itemize}
\end{lemma}
\begin{proof}[Of Proposition \ref{p:t-c-ub2}] Let $\C\neq \{0,1\}^n$ be a maximum class of VC dimension $0<d<n$. 
Fix $X \subseteq [n]$ of size $|X| =d+1$.
Because $\C$ is maximum,
the set $X$ is minimal non-shattered. Let $t$ be its unique forbidden trace. Define a family of traces $t_x\colon X\rightarrow \{0,1\}$ for $x\in X$, as $t_x(y) = t(y)$ for $y\neq x$ and $t_x(x) = 1 - t(x)$. 
For each $x \in X$,
let $c_x\in \C$ be a concept witnessing trace $t_x$. 

We claim that the family $\{c_x: x\in X\}$ is Radon independent,
so that $\radon(\C)\geq d+1$.
Indeed, let $(I,J)$ be a non-trivial partition of $X$. 
By choice of $(c_i : i \in I)$,
$$\conv_\C(\{c_i: i\in I\}) \subseteq \C_{J, t|_J}$$ where $t|_J$ is a restriction of $t$ to $J$.
Similarly, $\conv_\C(\{c_j: j\in J\}) \subseteq \C_{I, t|_I}$. Hence, 
$$\conv_\C(\{c_i: i\in I\}) \cap \conv_\C(\{c_j: j\in J\}) \subseteq \C_{J, t_J} \cap \C_{I, t_I} = \C_{X, t}.$$ But $\C_{X, t}=\emptyset$ because $t$ is a forbidden trace of $\C$ on $X$.
\end{proof}

\begin{proposition}[Example~\ref{ex:ball} - Dented cube]\label{p:ex-ball}
Let $\C=\{0,1\}^{d+1} \setminus \{(1,1,\ldots,1)\}$ be the dented cube. Then,
    \begin{enumerate}
	    \item \(\vc(\C)=d\),
	    \item $\vc^\star(\C)=\lfloor\log(d+1)\rfloor$,
	    \item $\radon(\C)=d+1$.
    \end{enumerate}
\end{proposition}
\begin{proof}
The first item is trivial and the second follows from Proposition~\ref{p:ex-cube}.
We only prove $\radon(\C)= d+1$.
Proposition~\ref{p:t-c-ub2} says that $\radon(\C)\geq d+1$. So we only need to show that $\radon(\C)\leq d+1$. Suppose, towards contradiction, that there are $d+2$ Radon shattered concepts $c_1, \dots, c_{d+2}\in \C$. Let us define partitions $I_i = \{c_i\}$ and $J_i = \{c_j: j\in [d+2] - \{i\}\}$ 
for $i\in [d+2]$. 
Write $\conv_\C(J_i) = \C_{X_i, t_i}$
for some $X_i \subseteq [n]$
and $t_i : X_i \to \{0,1\}$.
Because
$$\emptyset = \conv_\C(I_i) \cap \conv_\C(J_i) 
= I_i \cap  \C_{X_i, t_i}$$
there is $x_i\in X_i$ such that $c_i(x_i) \neq t_i(x_i)$
and $c_j(x_i) = t_i(x_i)$ for all $j \neq i$.
But this implies that the map $[d+2] \ni i \mapsto x_i \in [d+1]$ is one-to-one, which is a contradiction.
\end{proof}

\begin{proposition}[Example~\ref{ex:d=1}]\label{p:ex-d=1}
    For the class 
    \begin{align*}
\C  = \{01010101, & \allowbreak 11010101, \allowbreak \bb{10010101}, 
 		\allowbreak 01110101, \allowbreak \\ & \bb{01100101}, 
 		\allowbreak 01011101, \allowbreak \bb{01011001}, 
 		\allowbreak 01010111, \allowbreak \bb{01010110}\} ,
\end{align*}
 it holds that
    \begin{enumerate}
    \item $\vc(\C)=1$,
    \item $\vc^\star(\C)=3$ (any three out of the four blue concepts are dually shattered), 
    \item $\radon(\C)=3$.
    \end{enumerate}
\end{proposition}
\begin{proof}
	It can be verified that $\vc(\C)= 1$ and that it indeed dually shatters any three out of the four blue concepts, witnessing $\vc^*(\C)\geq 3$. This analysis can be simplified by noticing that $\C$ is highly symmetric: it remains the same under permutations that exchange the pairs of even and odd vertices, for example, $01234567 \mapsto 23671045$.
The class $\C$ is maximum because $|\C| = 9 = {8 \choose \leq 1}$. Theorem~\ref{t:c} implies that $3 \leq \vc^\star(\C)\leq \radon(\C) \leq 2\vc(\C)+1 = 3$, and so $\vc^\star(\C)= \radon(\C)=3$, as claimed.
\end{proof}

\begin{proposition}[Theorem~\ref{t:d}]\label{p:t-d}
\begin{enumerate}
    \item 
Every concept class $\C$ satisfies $\vc^\star(\C)\leq \radon(\C)$. 
\item For general classes, $\radon(\C)$ is not upper-bounded by $\vc(\C)$ or $\vc^*(\C)$. 
\item If $\C$ is extremal, then $\radon(\C) \leq 2^{\vc^\star(\C)+2}-1$, and this bound is tight up to a multiplicative factor.
\end{enumerate}
\end{proposition}

\begin{proof}
The fact that $\vc^\star(\C)\leq \radon(\C)$ was already proved in the discussion following Theorem~\ref{t:c}. The fact that $\radon(\C)$ is, in general, unbounded with respect to $\vc(\C)$ or $\vc^*(\C)$ is demonstrated by the class of singletons (Example~\ref{ex:singletons}).
For the third item, 
let $\C$ be extremal.  Theorem~\ref{t:c} says that $\radon(\C) \leq 2\vc(\C) + 1$. Theorem~\ref{t:assouad} says that $\vc(\C) \leq 2^{\vc^*(\C)+1} - 1$. Thus, 
$$\radon(\C) \leq 2\vc(\C) + 1 \leq 2(2^{\vc^*(\C)+1} - 1) + 1 = 2^{\vc^*(\C)+2} - 1.$$
The tightness up to a factor of $2$ is demonstrated by the dented cube (Example~\ref{ex:ball}) where $d+1=2^k - 1$. In this case, $\radon(\C) = 2^k - 1$ and $\vc^*(\C) = \lfloor \log(2^k -1)\rfloor = k-1$.
\end{proof}

\section{Discussion of the Main Proof}\label{sec:proof-discussion}
We conclude with a discussion of the Radon number, which is an interesting concept that merits a discussion in its own right. Most of the relevant concepts have been introduced in the main proof (\Cref{sec:overview}).

The classical Radon Theorem asserts that 
the maximal size of a Radon-independent set in \(\R^d\) is \(d+1\) with respect to the usual convexity structure. Equivalently, it says that
the maximal $k$ for which there is a \emph{linear} Radon map $f\colon \partial\Delta^{k+1}\rightarrow \R^d$ is $d-1$. The Topological Radon Theorem extends it by allowing \(f\) to be any continuous function. This opens the way to defining Radon numbers for arbitrary topological spaces, as we no longer depend on neither convexity nor linearity.

Other concepts can be formulated similarly. For example, we can define the \emph{Borsuk-Ulam number} $\bu(X)$ of a topological space $X$ as a maximal $d$ such that there is a Borsuk-Ulam map $f\colon S^d\rightarrow X$;
that is, a continuous map that does not collapse antipodal points. 
This way, \Cref{prop:guilbault}  proves that for all topological spaces $X$,
$$\radon(X)\leq\bu(X).$$ The Topological Radon theorem ($\radon(\R^d)=d-1$) becomes a consequence of the Borsuk-Ulam theorem ($\bu(\R^d)=d-1$), and Lemma~\ref{lem:radon-for-sc} ($\radon(K)\leq 2d-1$) a consequence of Theorem~\ref{thm:Jaw} ($\bu(K)\leq 2d-1$). It is natural to ask if the opposite direction holds.

\begin{open}
\label{q:4}
    Is it true that for an arbitrary topological space $X$, we have $\radon(S)=\bu(S)$?
    Is it true for the special case of simplicial complexes?
\end{open}

The question can be thought of as a rigidity property of antipodality. A positive answer says that the existence of a Borsuk-Ulam map implies the existence of a Radon map. The Borsuk-Ulam condition just requires that there is no collapse of antipodal points. The Radon condition requires e.g.\ that the image of a point is disjoint from all of its ``opposite faces''. 

In the one-dimensional case,
the answer seems to be positive. 
For the special case of simplicial complexes, 
if there is a Borsuk-Ulam map from the circle $S^1$ to a graph, 
then the graph is not a path (or a collection of disjoint paths).
Then in the graph, there is either a cycle or a vertex of degree at least three.
In both cases, there is a Radon map from $\partial \Delta^2$ to this graph. 
For general topological spaces $X$, consider the path space of $X$. Namely, the space of continuous maps $[0,1]\to X$.
If there are two paths whose union is connected but is not a path then there is a Radon map
(there is a point of ``degree three''). Otherwise, the space $X$ itself is a path, which is a contradiction to the existence of a Borsuk-Ulam map.

The argument for the one-dimensional case is different than the proof that a Radon map yields a Borsuk-Ulam map using
\Cref{prop:guilbault}.
Even for the one-dimensional case, there is no proof of this kind. That is, there is a one-dimensional space $X$ and a Borsuk-Ulam map $f:S^1 \to X$
so that for every continuous map
$g : \partial \Delta^2 \to S^1$,
the map $g \circ f$ is not Radon. 
The space $X$ can be chosen to be the circle $S^1$, and the map $f$ can be chosen to be the map that wraps around three times. 

In the Topological Radon Theorem, we can deal with two types of maps: $\partial \Delta^{d+1} \to X$ 
and $\Delta^{d+1}\to X$.
We used the first option,
but it also makes sense to use
the second.
If we denote this modified version of the Radon number by $\radon'(X)$, then $\radon'(X)\leq \radon(X)$ for all $X$.
For Euclidean space, we have $\radon'(\R^d) = \radon(\R^d)$ for all $d$, which could explain why this distinction is not made in the context of the Topological Radon Theorem. However, while $\radon(S^d) = \bu(S^d) = d$ as witnessed by the identity map, it seems that proving $\radon'(S^d)=d-1$ should not be too hard.

We were not able to find an \emph{explicit} definitions of topological Radon or Borsuk-Ulam numbers in the literature. They were, however, implicitly studied. For example, the aforementioned \cite{jaw:02} addressed $\bu(K)$, \cite{ja:93} studied the Borsuk-Ulam number of contractible simplicial complexes, and \cite{pergher:11} of products of manifolds. The latter paper, in fact, explicitly introduces a setup akin to (and even more general than) our definition of $\bu(S)$, and also has a good list of relevant references. The most systematic approach that we are aware of is developed in the book \cite{matousek} under the name of $\Z_2$-coindex of a space. 
This setup, nevertheless, is different because the target space is also required to have an involution, and the maps to be equivariant.

Finally, let us discuss the \emph{topological Radon number} of a class $\C$, defined as the Radon number of the cube complex $Q = Q(\C)$:
$$\radon_t(\C) = \radon(Q).$$ 
The bound 
$\radon(\C)\leq \radon_t(\C) + 2$ 
from Proposition~\ref{prop:embed} is specific to extremal classes. For example, in the (non-extremal) singletons class $\C$ in \Cref{ex:singletons}, the complex $Q$ is a collection of $n$ isolated points, and so $\radon_t(\C) = 0$ while $\radon(\C)=n$. 
The connection between $\radon$ and $\radon_t$ for extremal classes utilizes their topological nature in an essential way.

%


\bibliographystyle{plainnat}
\bibliography{learning}

\begin{thebibliography}{49}
\providecommand{\natexlab}[1]{#1}
\providecommand{\url}[1]{\texttt{#1}}
\expandafter\ifx\csname urlstyle\endcsname\relax
  \providecommand{\doi}[1]{doi: #1}\else
  \providecommand{\doi}{doi: \begingroup \urlstyle{rm}\Url}\fi

\bibitem[Alon et~al.(1992)Alon, B{\'a}r{\'a}ny, F{\"u}redi, and
  Kleitman]{alon:1992}
Noga Alon, Imre B{\'a}r{\'a}ny, Zolt{\'a}n F{\"u}redi, and Daniel~J Kleitman.
\newblock Point selections and weak $\varepsilon$-nets for convex hulls.
\newblock \emph{Combinatorics, Probability and Computing}, 1\penalty0
  (03):\penalty0 189--200, 1992.

\bibitem[Alon et~al.(2017)Alon, Moran, and Yehudayoff]{alon2017sign}
Noga Alon, Shay Moran, and Amir Yehudayoff.
\newblock Sign rank versus {V}apnik-{C}hervonenkis dimension.
\newblock \emph{Sbornik: Mathematics}, 208\penalty0 (12):\penalty0 1724, 2017.

\bibitem[Assouad(1983)]{assouad:83}
Patrick Assouad.
\newblock Densit\'e et dimension.
\newblock \emph{Annales de l'Institut Fourier (Grenoble)}, 33\penalty0
  (3):\penalty0 233--282, 1983.

\bibitem[Bajm{\'o}czy and B{\'a}r{\'a}ny(1979)]{Bajmoczy:79}
Ervin~G. Bajm{\'o}czy and Imre B{\'a}r{\'a}ny.
\newblock On a common generalization of {B}orsuk's and {R}adon's theorem.
\newblock \emph{Acta Mathematica Academiae Scientiarum Hungarica}, 34\penalty0
  (3):\penalty0 347--350, 1979.
\newblock \doi{10.1007/BF01896131}.
\newblock URL \url{https://doi.org/10.1007/BF01896131}.

\bibitem[B{\'a}r{\'a}ny et~al.(1990)B{\'a}r{\'a}ny, F{\"u}redi, and
  Lov{\'a}sz]{barany:1990}
Imre B{\'a}r{\'a}ny, Zolt{\'a}n F{\"u}redi, and L{\'a}szl{\'o} Lov{\'a}sz.
\newblock On the number of halving planes.
\newblock \emph{Combinatorica}, 10\penalty0 (2):\penalty0 175--183, 1990.

\bibitem[Ben{-}David and Litman(1998)]{Ben-David:98}
Shai Ben{-}David and Ami Litman.
\newblock {Combinatorial Variability of {Vapnik-Chervonenkis} Classes with
  Applications to Sample Compression Schemes}.
\newblock \emph{Discrete Applied Mathematics}, 86\penalty0 (1):\penalty0 3--25,
  1998.
\newblock \doi{10.1016/S0166-218X(98)00000-6}.
\newblock URL \url{http://dx.doi.org/10.1016/S0166-218X(98)00000-6}.

\bibitem[Bollob{\'a}s and Radcliffe(1995)]{Bollobas:95}
B{\'e}la Bollob{\'a}s and Andrew~J. Radcliffe.
\newblock Defect {S}auer results.
\newblock \emph{Journal of Combinatorial Theory, Series A}, 72\penalty0
  (2):\penalty0 189--208, 1995.

\bibitem[Bollob{\'a}s et~al.(1989)Bollob{\'a}s, Leader, and
  Radcliffe]{Bollobas:89}
B{\'e}la Bollob{\'a}s, Imre Leader, and Andrew~J. Radcliffe.
\newblock {Reverse Kleitman Inequalities}.
\newblock \emph{Proceedings of the London Mathematical Society}, s3-58\penalty0
  (1):\penalty0 153--168, 01 1989.
\newblock ISSN 0024-6115.
\newblock \doi{10.1112/plms/s3-58.1.153}.
\newblock URL \url{https://doi.org/10.1112/plms/s3-58.1.153}.

\bibitem[Bousquet et~al.(2020)Bousquet, Hanneke, Moran, and
  Zhivotovskiy]{bousquet:20}
Olivier Bousquet, Steve Hanneke, Shay Moran, and Nikita Zhivotovskiy.
\newblock Proper learning, {H}elly number, and an optimal {SVM} bound.
\newblock In \emph{Proceedings of the $33^{{\rm rd}}$ Conference on Learning
  Theory}, 2020.

\bibitem[Chalopin et~al.(2022)Chalopin, Chepoi, Moran, and
  Warmuth]{Chalopin:22}
J{\'{e}}r{\'{e}}mie Chalopin, Victor Chepoi, Shay Moran, and Manfred~K.
  Warmuth.
\newblock Unlabeled sample compression schemes and corner peelings for ample
  and maximum classes.
\newblock \emph{J. Comput. Syst. Sci.}, 127:\penalty0 1--28, 2022.
\newblock \doi{10.1016/J.JCSS.2022.01.003}.
\newblock URL \url{https://doi.org/10.1016/j.jcss.2022.01.003}.

\bibitem[Chalopin et~al.(2023)Chalopin, Chepoi, Inerney, Ratel, and
  Vax{\`{e}}s]{Chalopin:23}
J{\'{e}}r{\'{e}}mie Chalopin, Victor Chepoi, Fionn~Mc Inerney, S{\'{e}}bastien
  Ratel, and Yann Vax{\`{e}}s.
\newblock Sample compression schemes for balls in graphs.
\newblock \emph{{SIAM} J. Discret. Math.}, 37\penalty0 (4):\penalty0
  2585--2616, 2023.
\newblock \doi{10.1137/22M1527817}.
\newblock URL \url{https://doi.org/10.1137/22m1527817}.

\bibitem[Chepoi et~al.(2020)Chepoi, Knauer, and Philibert]{chepoi:20}
Victor Chepoi, Kolja Knauer, and Manon Philibert.
\newblock Two-dimensional partial cubes.
\newblock \emph{The Electronic Journal of Combinatorics}, pages 3--29, 2020.

\bibitem[Chepoi et~al.(2021)Chepoi, Knauer, and Philibert]{Chepoi:21}
Victor Chepoi, Kolja Knauer, and Manon Philibert.
\newblock Labeled sample compression schemes for complexes of oriented
  matroids.
\newblock \emph{CoRR}, abs/2110.15168, 2021.
\newblock URL \url{https://arxiv.org/abs/2110.15168}.

\bibitem[Chepoi et~al.(2022)Chepoi, Knauer, and Philibert]{chepoi:22}
Victor Chepoi, Kolja Knauer, and Manon Philibert.
\newblock Ample completions of oriented matroids and complexes of uniform
  oriented matroids.
\newblock \emph{{SIAM} Journal of Discrete Mathematics}, 36\penalty0
  (1):\penalty0 509--535, 2022.

\bibitem[Chornomaz(2022)]{chornomaz:22}
Bogdan Chornomaz.
\newblock What convex geometries tell about shattering-extremal systems.
\newblock \emph{The Electronic Journal of Combinatorics}, pages P3--40, 2022.

\bibitem[Danzer et~al.(1963)Danzer, Gr{\"u}nbaum, and Klee]{danzer1963helly}
Ludvig Danzer, Branko Gr{\"u}nbaum, and Victor Klee.
\newblock \emph{Helly's Theorem and Its Relatives}.
\newblock Proceedings of symposia in pure mathematics: Convexity. American
  Mathematical Society, 1963.
\newblock URL \url{https://books.google.com/books?id=I1l5HAAACAAJ}.

\bibitem[Dress(1996)]{Dress:96}
Andreas W.~M. Dress.
\newblock Towards a theory of holistic clustering.
\newblock In Boris~G. Mirkin, Fred~R. McMorris, Fred~S. Roberts, and Andrey
  Rzhetsky, editors, \emph{Mathematical Hierarchies and Biology, Proceedings of
  a {DIMACS} Workshop, November 13-15, 1996}, volume~37 of \emph{{DIMACS}
  Series in Discrete Mathematics and Theoretical Computer Science}, pages
  271--290. {DIMACS/AMS}, 1996.
\newblock \doi{10.1090/DIMACS/037/19}.
\newblock URL \url{https://doi.org/10.1090/dimacs/037/19}.

\bibitem[Floyd and Warmuth(1995)]{floyd:95}
Sally Floyd and Manfred Warmuth.
\newblock Sample compression, learnability, and the {V}apnik-{C}hervonenkis
  dimension.
\newblock \emph{Machine Learning}, 21\penalty0 (3):\penalty0 269--304, 1995.

\bibitem[G{\"{a}}rtner and Welzl(1994)]{Gartner:94}
Bernd G{\"{a}}rtner and Emo Welzl.
\newblock {V}apnik-{C}hervonenkis dimension and (pseudo-) hyperplane
  arrangements.
\newblock \emph{Discrete Comput. Geom.}, 12\penalty0 (4):\penalty0 399--432,
  1994.
\newblock \doi{10.1007/BF02574389}.

\bibitem[Gon{\c{c}}alves et~al.(2002)Gon{\c{c}}alves, Jaworowski, and
  Pergher]{jaw:02}
Daciberg Gon{\c{c}}alves, Jan Jaworowski, and Pedro Pergher.
\newblock {$G$}-coincidences for maps of homotopy spheres into {CW}-complexes.
\newblock \emph{Proceedings of the American Mathematical Society}, 130\penalty0
  (10):\penalty0 3111--3115, 2002.

\bibitem[Guilbault(2010)]{guilbault:10}
Craig~R. Guilbault.
\newblock An elementary deduction of the {T}opological {R}adon {T}heorem from
  {B}orsuk--{U}lam.
\newblock \emph{Discrete \& Computational Geometry}, 43:\penalty0 951--954,
  2010.

\bibitem[Helly(1923)]{Helly:23}
Eduard Helly.
\newblock Über mengen konvexer körper mit gemeinschaftlichen punkte.
\newblock \emph{Jahresbericht der Deutschen Mathematiker-Vereinigung},
  32:\penalty0 175--176, 1923.
\newblock URL \url{http://eudml.org/doc/145659}.

\bibitem[Holmsen and Lee(2021)]{holmsen:2021}
Andreas~F Holmsen and Donggyu Lee.
\newblock Radon numbers and the fractional {H}elly theorem.
\newblock \emph{Israel Journal of Mathematics}, 241\penalty0 (1):\penalty0
  433--447, 2021.

\bibitem[Jaworowski and Izydorek(1993)]{ja:93}
Jan Jaworowski and Marek Izydorek.
\newblock Antipodal coincidence for maps of spheres into complexes.
\newblock In \emph{preprint (Lecture at the Conference on Topological Fixed
  Point Theory and its Applications in Torun}, 1993.

\bibitem[Kay and Womble(1971)]{kay1971axiomatic}
David Kay and Eugene~W Womble.
\newblock Axiomatic convexity theory and relationships between the
  {C}arath{\'e}odory, {H}elly, and {R}adon numbers.
\newblock \emph{Pacific Journal of Mathematics}, 38\penalty0 (2):\penalty0
  471--485, 1971.

\bibitem[Kuzmin and Warmuth(2007)]{Kuzmin:07}
Dima Kuzmin and Manfred~K. Warmuth.
\newblock Unlabeled compression schemes for maximum classes.
\newblock \emph{Journal of Machine Learning Research}, 8:\penalty0 2047--2081,
  2007.
\newblock URL \url{http://dl.acm.org/citation.cfm?id=1314566}.

\bibitem[Lawrence(1983)]{lawrence:83}
James~F. Lawrence.
\newblock Lopsided sets and orthant-intersection of convex sets.
\newblock \emph{Pacific J. Math.}, 104:\penalty0 155--173, 1983.

\bibitem[Lee and Santos(2017)]{triangulations}
Carl~W Lee and Francisco Santos.
\newblock Subdivisions and triangulations of polytopes.
\newblock In \emph{Handbook of discrete and computational geometry}, pages
  415--447. Chapman and Hall/CRC, 2017.

\bibitem[Levi(1951)]{levi1951helly}
Friedrich~W. Levi.
\newblock On {H}elly's theorem and the axioms of convexity.
\newblock \emph{J. Indian Math. Soc}, 15:\penalty0 65--76, 1951.

\bibitem[Littlestone and Warmuth(1986)]{littlestone:86}
Nick Littlestone and Manfred~K. Warmuth.
\newblock Relating data compression and learnability.
\newblock \emph{Unpublished manuscript}, 1986.

\bibitem[Matou{\v{s}}ek et~al.(2003)Matou{\v{s}}ek, Bj{\"o}rner, Ziegler,
  et~al.]{matousek}
Ji{\v{r}}{\'\i} Matou{\v{s}}ek, Anders Bj{\"o}rner, G{\"u}nter~M Ziegler,
  et~al.
\newblock \emph{Using the {B}orsuk-{U}lam theorem: lectures on topological
  methods in combinatorics and geometry}, volume 2003.
\newblock Springer, 2003.

\bibitem[Moran(2012)]{Moran:12}
Shay Moran.
\newblock Shattering-extremal systems.
\newblock \emph{arXiv preprint}, 1211.2980, 2012.

\bibitem[Moran and Warmuth(2016)]{moran2016labeled}
Shay Moran and Manfred Warmuth.
\newblock Labeled compression schemes for extremal classes.
\newblock In \emph{Proceedings of the $27^{{\rm th}}$ International Conference
  on Algorithmic Learning Theory}, pages 34--49. Springer, 2016.

\bibitem[Moran and Yehudayoff(2016)]{moranyehudayoof:16}
Shay Moran and Amir Yehudayoff.
\newblock Sample compression schemes for {VC} classes.
\newblock \emph{Journal of the {ACM}}, 63\penalty0 (3):\penalty0 1--10, 2016.

\bibitem[Moran and Yehudayoff(2020)]{MoranY:20}
Shay Moran and Amir Yehudayoff.
\newblock On weak {\(\epsilon\)}-nets and the {R}adon number.
\newblock \emph{Discret. Comput. Geom.}, 64\penalty0 (4):\penalty0 1125--1140,
  2020.
\newblock \doi{10.1007/S00454-020-00222-Y}.
\newblock URL \url{https://doi.org/10.1007/s00454-020-00222-y}.

\bibitem[Pajor(1985)]{Pajor:1985}
Alain Pajor.
\newblock \emph{Sous-espaces $\ell_1^n$ des Espaces de {B}anach}.
\newblock Travaux en Cours. Hermann, Paris, 1985.

\bibitem[Radon(1921)]{Radon:21}
Johann Radon.
\newblock Mengen konvexer k{\"o}rper, die einen gemeinsamen punkt enthalten.
\newblock \emph{Mathematische Annalen}, 83\penalty0 (1):\penalty0 113--115, Mar
  1921.
\newblock ISSN 1432-1807.
\newblock \doi{10.1007/BF01464231}.
\newblock URL \url{https://doi.org/10.1007/BF01464231}.

\bibitem[Rubinstein and Rubinstein(2012)]{Rubinstein:12}
Benjamin I.~P. Rubinstein and J.~Hyam Rubinstein.
\newblock A geometric approach to sample compression.
\newblock \emph{Journal of Machine Learning Research}, 13:\penalty0 1221--1261,
  2012.
\newblock URL \url{http://dl.acm.org/citation.cfm?id=2343686}.

\bibitem[Rubinstein et~al.(2009)Rubinstein, Bartlett, and
  Rubinstein]{Rubinstein:09}
Benjamin I.~P. Rubinstein, Peter~L. Bartlett, and J.~Hyam Rubinstein.
\newblock Shifting: One-inclusion mistake bounds and sample compression.
\newblock \emph{J. Comput. Syst. Sci.}, 75\penalty0 (1):\penalty0 37--59, 2009.
\newblock \doi{10.1016/J.JCSS.2008.07.005}.
\newblock URL \url{https://doi.org/10.1016/j.jcss.2008.07.005}.

\bibitem[Rubinstein et~al.(2015)Rubinstein, Rubinstein, and
  Bartlett]{rubinstein:15}
J.~Hyam Rubinstein, Benjamin I.~P. Rubinstein, and Peter~L Bartlett.
\newblock Bounding embeddings of {VC} classes into maximum classes.
\newblock \emph{Measures of Complexity: {F}estschrift for {A}lexey
  {C}hervonenkis}, pages 303--325, 2015.

\bibitem[Rubinstein and Rubinstein(2022)]{Rubinstein:22}
Joachim~Hyam Rubinstein and Benjamin I.~P. Rubinstein.
\newblock Unlabelled sample compression schemes for intersection-closed classes
  and extremal classes.
\newblock In \emph{NeurIPS}, 2022.
\newblock URL
  \url{http://papers.nips.cc/paper\_files/paper/2022/hash/54d6a55225cebbdc16fbb0e45c5bdf2b-Abstract-Conference.html}.

\bibitem[Sauer(1972)]{sauer:72}
Norbert Sauer.
\newblock On the density of families of sets.
\newblock \emph{Journal of Combinatorial Theory (A)}, 13\penalty0 (1):\penalty0
  145--147, 1972.

\bibitem[Shalev-Shwartz and Ben-David(2014)]{Shaelv-Shwartz:book}
Shai Shalev-Shwartz and Shai Ben-David.
\newblock \emph{Understanding Machine Learning - From Theory to Algorithms.}
\newblock Cambridge University Press, 2014.
\newblock ISBN 978-1-10-705713-5.

\bibitem[Shelah(1972)]{Shelah:72}
Saharon Shelah.
\newblock A combinatorial problem, stability and order for models and theories
  in infinitary languages.
\newblock \emph{Pacific J. Math.}, 41\penalty0 (1):\penalty0 247--261, 1972.
\newblock \doi{10.2140/pjm.1972.41.247}.

\bibitem[van~de Vel(1993)]{van1993theory}
M.L.J. van~de Vel.
\newblock \emph{Theory of Convex Structures}, volume~50 of \emph{North-Holland
  mathematical library}.
\newblock North-Holland, 1993.
\newblock ISBN 9780444815057.
\newblock URL \url{https://books.google.com/books?id=xt9-lAEACAAJ}.

\bibitem[Vapnik and Chervonenkis(1968)]{vapnik:68}
Vladimir~N. Vapnik and Alexey~Ya. Chervonenkis.
\newblock On the uniform convergence of relative frequencies of events to their
  probabilities.
\newblock \emph{Proc. USSR Acad. Sci.}, 181\penalty0 (4):\penalty0 781--783,
  1968.

\bibitem[Vendruscolo et~al.(2011)Vendruscolo, Desideri, Pergher,
  et~al.]{pergher:11}
Daniel Vendruscolo, Patricia~E Desideri, Pedro~LQ Pergher, et~al.
\newblock Some generalizations of the {B}orsuk-{U}lam theorem.
\newblock \emph{Plubl. Math. Debrecen}, 78:\penalty0 583--593, 2011.

\bibitem[Warmuth(2003)]{Warmuth:03}
Manfred~K. Warmuth.
\newblock Compressing to {VC} dimension many points.
\newblock In \emph{COLT/Kernel}, pages 743--744, 2003.
\newblock \doi{10.1007/978-3-540-45167-9_60}.
\newblock URL \url{http://dx.doi.org/10.1007/978-3-540-45167-9_60}.

\bibitem[Wigderson(2019)]{Wigderson:17}
Avi Wigderson.
\newblock \emph{A Theory Revolutionizing Technology and Science}.
\newblock Princeton University Press, Princeton, 2019.
\newblock ISBN 9780691192543.
\newblock \doi{doi:10.1515/9780691192543}.
\newblock URL \url{https://doi.org/10.1515/9780691192543}.

\end{thebibliography}

\end{document}